\documentclass[conference]{IEEEtran}
\usepackage{graphicx}
\usepackage{float}
\usepackage{caption}
\usepackage{multicol}
\usepackage{mathrsfs}

\usepackage{amsmath}
\usepackage{amssymb}
\usepackage{amsthm}
\usepackage{multicol}
\usepackage{mathtools}
\usepackage{float}
\usepackage{placeins}
\usepackage{epstopdf}
\usepackage{multirow}
\usepackage{subfigure}
\usepackage{afterpage}
\usepackage{pdflscape}
\usepackage{float}
\restylefloat{table}
\let\labelindent\relax
\usepackage{enumitem}
\usepackage[utf8]{inputenc}
\floatstyle{ruled}
\newfloat{algorithm}{tbp}{loa}
\providecommand{\algorithmname}{Algorithm}
\floatname{algorithm}{\protect\algorithmname}
\theoremstyle{remark}
\newtheorem{theorem}{Theorem}

\theoremstyle{remark}
\newtheorem{example}{Example}

\title{Optimal Scalar Linear Index Codes for Some Symmetric Multiple Unicast Problems}
\begin{document}
\author{Mahesh Babu Vaddi, Roop Kumar Bhattaram and B.~Sundar~Rajan, {\it Fellow,~IEEE}
}
\maketitle
\begin{abstract}
The capacity of symmetric instance of the multiple unicast index coding problem with neighboring antidotes (side-information) with number of messages equal to the number of receivers was given by Maleki, Cadambe and Jafar in \cite{MCJ}.  In this paper we consider ten symmetric multiple unicast problems with lesser antidotes than considered in \cite{MCJ} and explicitly construct scalar linear codes for them. These codes are shown to achieve the capacity or equivalently these codes shown to be of optimal length.  Also, the constructed codes enable the receivers use small number of transmissions to decode their wanted messages which is important to have the probability of message error reduced in a noisy broadcast channel \cite{TRCR}, \cite{AnR}. Some of the cases considered are shown to be critical index coding problems and these codes help to identify some of the subclasses considered in \cite{MCJ} to be not critical index coding problems. \footnote{The authors are with the Department of Electrical Communication Engineering, Indian Institute of Science, Bangalore-560012, India. Email:bsrajan@ece.iisc.ernet.in.} 
\end{abstract}
\section{Introduction}
\label{sec:Introduction}
\IEEEPARstart {T}{he} problem of index coding with side information was introduced by Birk and Kol \cite{BiK}. Bar-Yossef \textit{et al.} \cite{YBJK} studied a type of index coding problem in which each receiver demands only one single message and the number of receivers equals number of messages. Ong and Ho \cite{OnH} classify the binary index coding problem depending on the demands and the side information possessed by the receivers. An index coding problem is unicast if the demand sets of the receivers are disjoint. If the problem is unicast and if the size of each demand set is one, then it is said to be single unicast. It was found that the length of the optimal linear index code is equal to the minrank of the side information graph of the index coding problem but finding the minrank is NP hard. \\

Recently, it has been observed that in a noisy index coding problem it is desirable for the purpose of reducing the probability of error that  the receivers use as small a number of transmissions from the source as possible and linear index codes with this property have been reported in \cite{TRCR}, \cite{KaR}. While the report \cite{TRCR} considers fading broadcast channels, in \cite{AnR} AWGN channels are considered and it is reported that linear index codes with minimum length (capacity achieving codes or optimal length codes) help to facilitate to achieve more reduction in probability of error compared to non-minimum length codes for receivers with large amount of side-information.\\

Maleki \textit{et al.} \cite{MCJ} found the capacity of symmetric multiple unicast index problem with neighboring antidotes (side information). In a symmetric multiple unicast index coding problem with equal number of $K$  messages and source-destination pairs, each destination has a total of $U+D=A<K$ antidotes, corresponding to the $U$ messages before (``up" from) and $D$ messages after (``down" from) its desired message. In this setting, the $k-$th receiver $R_{k}$ demands the message $x_{k}$ having the antidotes 
\begin{equation}
\label{antidote}
\{x_{k-U},\dots,x_{k-2},x_{k-1}\}~\cup~\{x_{k+1}, x_{k+2},\dots,x_{k+D}\}.
\end{equation}
The symmetric capacity of this index coding problem setting is shown to be as follows:
\begin{flushleft}
$U,D \in$ $\mathbb{Z},$\\
$0 \leq U \leq D$,\\
$U+D=A<K$,\ is\\
$C=\left\{
                \begin{array}{ll}
                  {1,\qquad\quad\ A=K-1}\\
                  {\frac{U+1}{K-A+2U}},A\leq K-2\qquad $per message.$
                  \end{array}
              \right.$
\end{flushleft}
\ \\
The above expression for capacity per message can be expressed as below for arbitrary $U$ and $D$:
\begin{equation}
\label{capacity}
C=\left\{
                \begin{array}{ll}
                  {1 ~~~~~~~~~~~~~~~~~~~~~~~~~~~~~ \mbox{if} ~~ U+D=K-1}\\
                  {\frac{min(U,D)+1}{K+min(U,D)-max(U,D)}} ~~~ \mbox{if} ~~U+D\leq K-2. 
                  \end{array}
              \right.
\end{equation}
The side information in the above index coding problem is represented by a directed graph G = ($V$,$E$) with $V = \{1,2,...,K\}$ is the set of vertices and E is the set of edges such that the directed edge $(i,j)\in E$ if receiver (destination) $R_{i}$ knows $x_{j}$. This graph $G$ for a given index coding problem is called side information graph. Let $G$ be a directed graph of $K$ vertices without self loops. A 0-1 matrix $A=(a_{i,j})$ fits in $G$ 
if $a_{i,i}=1$ for all $i$ and $a_{i,j}$=0 whenever $(i,j)$ is not an edge of $G$. Let $rk_{2}()$ denotes the rank of the 0-1 matrix over $GF(2)$. The $minrank_{2}(G)$ is defined as \cite{YBJK}
$$minrank_{2}(G) \triangleq min\{rk_{2}(A) : A \ fits \ in \ G\}.$$
In a given index coding problem with side information graph $G$ = $(V,E)$, an edge $e\in E$ is said to be critical if the removal of $e$ from $G$ strictly reduce the capacity. The index coding problem $G=(V,E)$ is critical if every $e\in E$ is critical \cite{TSG}.  

In the setting of \cite{MCJ} with one sided antidote cases, i.e., the cases where $U$ or $D$ is zero, 
without loss of generality, we can assume that $max(U,D)= D$ and $min(U,D)=0$ (all the results hold when $max(U,D)=U$), i.e., 
\begin{equation}
\label{antidote1}
{\cal K}_k =\{x_{k+1}, x_{k+2},\dots,x_{k+D}\}, 
\end{equation}
\noindent
for which \eqref{capacity} reduces to 
\begin{equation}
\label{capacity1}
C=\left\{
                \begin{array}{ll}
                  {1 ~~~~~~~~~~~~ \mbox{if} ~~ D=K-1}\\
                  {\frac{1}{K-D}} ~~~~~~ \mbox{if} ~~D\leq K-2. 
                  \end{array}
              \right.
\end{equation}
symbols per message. 

\subsection{Contributions}
In this paper we consider the following ten cases of symmetric multiple unicast problems which are subclasses of the problems discussed in \cite{MCJ} with one sided antidotes: \\
\noindent
{\bf Case I:} ~~ $D$ divides $K$ and ${\cal K}_k =\{x_{k+D}\}$. \\
{\bf Case II:}~ $K-D$ divides $K$ and $${\cal K}_k=\{x_{k+K-D},x_{k+2(K-D)},...,x_{k+D}\}.$$
{\bf Case III:} $D-\frac{K}{2}$ divides $\frac{K}{2} $ and \\
 $${\cal K}_k=\{x_{k+\frac{K}{2}},x_{k+D-\frac{K}{2}},x_{k+D}\}.$$
{\bf Case IV:}~  $\frac{K}{2}-D$ divides $D$ and $${\cal K}_k=\{x_{k+\frac{K}{2}-D},x_{k+2(\frac{K}{2}-D)}+ \cdots +x_{k+D}\}.$$
{\bf Case V:}~~There is an integer $\lambda$ such that $D$ divides $K-\lambda$ and $\lambda$ divides $D$ and 
$${\cal K}_k=\left\{
                \begin{array}{ll}
                  \{x_{k+D}\},\ $if$\ k\leq K-D-\lambda\\
                  \{x_{k+\lambda},x_{k+2\lambda}\dots,x_{k+D}\},\ $if$\ K-D-\lambda<k\leq K
                  \end{array}
              \right.
$$
{\bf Case VI:}~There is an integer $\lambda$ such that $K-D$ divides $K-\lambda$ and $\lambda$ divides $K-D$ and $\mathcal{K}_k$ is as in \eqref{antidote1}.\\ \\
{\bf Case VII:} There is an integer $\lambda$ such that $D+\lambda$ divides $K$ and $\lambda$ divides $D$ and 
$${\cal K}_k=\{x_{k+\lambda},x_{k+2\lambda},...,x_{k+D}\}.$$\\
{\bf Case VIII:} There is an integer $\lambda$ such that $K-D+\lambda$ divides $K$ and $\lambda$ divides $K-D$ and 
$${\cal K}_k=\left\{
                \begin{array}{ll}
                  \{x_{k+\lambda},\ x_{k+\lambda+(K-D)},\\x_{k+2\lambda+(K-D)},\ x_{k+2\lambda+2(K-D)},\\ \qquad \vdots \qquad \qquad \qquad \qquad \vdots\\x_{k+(p-1)\lambda+(p-2)(K-D)},\ x_{k+(p-1)\lambda+(p-1)(K-D)},\\x_{k+p\lambda+(p-1)(K-D)}\}\\
                 
                  \end{array}
              \right.
$$
where $p=\frac{K}{K-D+\lambda}$.\\ \\
{\bf Case IX:} There is an integer $\lambda$ such that $D$ divides $K+\lambda$ and $\lambda$ divides $D$ and 
$${\cal K}_k=\left\{
                \begin{array}{ll}
                  \{x_{k+D}\}$, if $\ k\leq K-2D+\lambda\\
                  \{x_{k+\lambda}, x_{k+2\lambda},...,x_{k+D}\}$, if $K-2D+\lambda<k\leq K
                  \end{array}
              \right.$$
{\bf Case X:} There is an integer $\lambda$ such that $K-D$ divides $K+\lambda$ and $\lambda$ divides $K-D$ and $\mathcal{K}_k$ is as in  \eqref{antidote1}.\\

We show that for all the ten cases listed above the capacity is given by \eqref{capacity1}. This done by way of presenting linear index codes of length $\frac{1}{K-D}$ for all the ten cases. Also we identify the critical index coding problems among these. \\

The codes proposed in this paper have the property that most of the receivers use a very small number of transmissions to decode their messages. We give the exact number of transmissions used by each of the receivers for all the ten cases. \\ 

Throughout, binary field is assumed but all the results can be easily extended to any finite field.\\
\section{Scalar linear codes for some specific antidote settings}
\label{sec2}
In this section we present linear index codes for all the ten cases listed in the previous section and also show that the proposed codes are capacity achieving. Recall that the $k-$th receiver $R_k$  wants the message $x_k$ and has the antidote  ${\cal K}_k.$  \\

A linear index code on $K$ message symbols $\{x_1, x_2, \cdots, x_K \}$ is given by $\underline{x}L = \underline{y}$ where $\underline{x}$=$[x_{1},x_{2}, \dots ,x_{K}],$ $L$ is a $K \times N$ matrix called the generator matrix of the code, and $\underline{y}$=$[y_{1},y_{2}, \dots ,y_{N}].$ The entries of the codeword $\{y_1, y_2, \cdots, y_N \}$ are also called code symbols. The code is also described by $\mathfrak{C}=\{y_{1},y_{2}, \dots ,y_{N}\}.$ \\

\noindent
{\bf Case I:} $D$ divides $K$ and ${\cal K}_k =\{x_{k+D}\}$
\begin{theorem}
\label{thm1}
If $D$ divides $K$ and ${\cal K}_k =\{x_{k+D}\},$ then the proposed  optimal length scalar linear code is 
{\small
$$\mathfrak{C}=\{x_{i+(j-1)D}+x_{i+jD}|~i=1,2,\dots,D,~ j=1,2,\dots, \frac{K}{D}-1\}.$$ 
}
Moreover, the problem is a critical index coding problem and the $minrank$ of the corresponding side information graph is $(K-D)$.
\end{theorem}
\begin{proof}
The code book has $K-D$ code symbols and we need to show that all the $K$ receivers get their wanted message using the code symbols.  From the construction of the code for every $k\leq{K-D}$, the code book consists of the code symbol $x_k+x_{k+D}.$ Since the symbol $x_{k+D}$ is the antidote of the $k^{th}$ receiver, it can decode its required symbol $x_k$ from $x_k+x_{k+D}.$ \\
For $k>{K-D}$, the code book consists of the symbols $x_{k}+x_{k-D},~x_{k-D}+x_{k-2D},~\dots,~x_{k\ mod\ D+D}+x_{k\ mod\ D}$. By adding these $\frac{K}{D}-1$ code symbols we get $x_k+x_{k\ mod\ D}$. As $x_{k\ mod\ D}$ is the antidote of $R_{k}$, the $k^{th}$ receiver can decode $x_k$. \\
Next we proceed to show that this linear coding scheme has capacity $\frac{1}{K-D}$ symbols per message. From the proposed coding scheme consisting of $K-D$ code symbols, the capacity is at least  $\frac{1}{K-D}$ symbols per message. However, with ${\cal K}_k =\{x_{k+1}, x_{k+2},\dots,x_{k+D}\},$  having  more side information the capacity is known to be $\frac{1}{K-D}$ \cite{MCJ}. If we reduce the number of antidotes the capacity either decreases or at most remain same.  Therefore the capacity for the proposed code is at most $\frac{1}{K-D}$. Hence the capacity is $C=\frac{1}{K-D}$. This also means that the proposed code is of optimal length.\\
If for any of the message, say $x_j,$ the corresponding antidote is not available, then $x_j$ needs to be a code symbol transmitted alone in addition to all other code symbols being transmitted. This will decrease the capacity and hence the problem is a critical index coding problem.\\
That the length of the proposed code is optimal means the minrank of the side information graph of the given problem is $K-D$.
\end{proof}


Theorem \ref{thm1} means that the original index coding problem of \cite{MCJ} with one sided antidote  as given in \eqref{antidote1} is not a critical index coding problem for the case when $D$ divides $K.$
 
By using the proposed code, \mbox{$K-D$} receivers can decode their wanted messages by using just one transmitted symbol. The remaining receivers decode their wanted messages by using \mbox{$\frac{K}{D}-1$} transmissions each.\\ 
\begin{example}
\label{ex1}
Consider the case $D=4,\ K=20$ and ${\cal K}_k = k+4$  for $k=1,2, \cdots, 20.$ The capacity of this index coding problem is $C=\frac{1}{K-D}=\frac{1}{16}$. The proposed code is \\ 
 $\mathfrak{C}_{1}$=$\{x_1+x_5, ~ x_{2}+x_{6}, ~ x_{3}+x_{7}, ~ x_{4}+x_{8}, ~ x_5+x_9, \\
~~~~~~~~ x_{6}+x_{10}, ~ x_{7}+x_{11}, ~ x_{8}+x_{12}, ~ x_9+x_{13}, ~ x_{10}+x_{14}, \\
~~~~~~~~~ x_{11}+x_{15}, ~ x_{12}+x_{16}, ~ x_{13}+x_{17}, ~ x_{14}+x_{18}, ~ x_{15}+x_{19},\\
~~~~~~~~ x_{16}+x_{20}\}$.\\
The code is given by $\mathfrak{C_{1}}$=$\underline{x}L_{1}$, where $\underline{x}$=$[x_{1},x_{2}, \dots ,x_{20}]$ and $L_{1}$ is the $20 \times 16$ matrix shown below:

{\tiny
$$L_{1} = \left[
\begin{array}{cccccccccccccccc}
   1 & 0 & 0 & 0 & 0 & 0 & 0 & 0 & 0 & 0 & 0 & 0 & 0 & 0 & 0 & 0\\
   0 & 1 & 0 & 0 & 0 & 0 & 0 & 0 & 0 & 0 & 0 & 0 & 0 & 0 & 0 & 0\\
   0 & 0 & 1 & 0 & 0 & 0 & 0 & 0 & 0 & 0 & 0 & 0 & 0 & 0 & 0 & 0\\
   0 & 0 & 0 & 1 & 0 & 0 & 0 & 0 & 0 & 0 & 0 & 0 & 0 & 0 & 0 & 0\\
   1 & 0 & 0 & 0 & 1 & 0 & 0 & 0 & 0 & 0 & 0 & 0 & 0 & 0 & 0 & 0\\
   0 & 1 & 0 & 0 & 0 & 1 & 0 & 0 & 0 & 0 & 0 & 0 & 0 & 0 & 0 & 0\\
   0 & 0 & 1 & 0 & 0 & 0 & 1 & 0 & 0 & 0 & 0 & 0 & 0 & 0 & 0 & 0\\
   0 & 0 & 0 & 1 & 0 & 0 & 0 & 1 & 0 & 0 & 0 & 0 & 0 & 0 & 0 & 0\\
   0 & 0 & 0 & 0 & 1 & 0 & 0 & 0 & 1 & 0 & 0 & 0 & 0 & 0 & 0 & 0\\
   0 & 0 & 0 & 0 & 0 & 1 & 0 & 0 & 0 & 1 & 0 & 0 & 0 & 0 & 0 & 0\\
   0 & 0 & 0 & 0 & 0 & 0 & 1 & 0 & 0 & 0 & 1 & 0 & 0 & 0 & 0 & 0\\
   0 & 0 & 0 & 0 & 0 & 0 & 0 & 1 & 0 & 0 & 0 & 1 & 0 & 0 & 0 & 0\\
   0 & 0 & 0 & 0 & 0 & 0 & 0 & 0 & 1 & 0 & 0 & 0 & 1 & 0 & 0 & 0\\
   0 & 0 & 0 & 0 & 0 & 0 & 0 & 0 & 0 & 1 & 0 & 0 & 0 & 1 & 0 & 0\\
   0 & 0 & 0 & 0 & 0 & 0 & 0 & 0 & 0 & 0 & 1 & 0 & 0 & 0 & 1 & 0\\
   0 & 0 & 0 & 0 & 0 & 0 & 0 & 0 & 0 & 0 & 0 & 1 & 0 & 0 & 0 & 1\\
   0 & 0 & 0 & 0 & 0 & 0 & 0 & 0 & 0 & 0 & 0 & 0 & 1 & 0 & 0 & 0\\
   0 & 0 & 0 & 0 & 0 & 0 & 0 & 0 & 0 & 0 & 0 & 0 & 0 & 1 & 0 & 0\\
   0 & 0 & 0 & 0 & 0 & 0 & 0 & 0 & 0 & 0 & 0 & 0 & 0 & 0 & 1 & 0\\
   0 & 0 & 0 & 0 & 0 & 0 & 0 & 0 & 0 & 0 & 0 & 0 & 0 & 0 & 0 & 1\\
  \end{array}
\right]$$
}

\noindent
For $1 \leq k \leq 16,$ the receiver $R_k$ can get $x_k$ since it has $x_{k+4}$ as antidote.  By adding the first four code symbols $R_{17}$  gets $x_{1}+x_{17}$ and since $x_{1}$ is the antidote for it $x_{17}$ is obtained. In the same manner by adding the $5^{th}, 6^{th}, 7^{th}$ and $8^{th}$ code symbols $R_{18}$ can get $x_{2}+x_{18}$ and it can decode $x_{18}$. By adding the $9^{th}, 10^{th}, 11^{th}$ and $12^{th}$ code symbols $R_{19}$ gets  $x_{19}$. Finally, by adding the $13^{th}, 14^{th}, 15^{th}$ and $16^{th}$ code symbols $R_{20}$ gets  $x_{20}$. To decode the first $K-D$ receivers use only one code symbol and the remaining four receivers use four code symbols.  \\
\end{example}

\noindent
{\bf Case II:} 
$K-D$ divides $K$ and 
\begin{equation}
\label{antidote2}
{\cal K}_k=\{x_{k+K-D},x_{k+2(K-D)},...,x_{k+D}\}
\end{equation}
\begin{theorem}
\label{thm2}
If \mbox{$K-D$} divides $K$ and the antidote pattern is given by \eqref{antidote2}, then the proposed code is $${\mathfrak{C}}=\{x_{i}+x_{i+m}+\dots+x_{i+(n-1)m}\ |\ {i = \{1,2, \dots, m\}\}}$$ where  $K-D=m$ and \mbox{$\frac{K}{K-D}=n$}. The  capacity is $C=\frac{1}{K-D}$ and $minrank$ of the corresponding side information graph for this antidote pattern is $K-D$.
\end{theorem}
\begin{proof}
We show that by using this code every receiver can decode its wanted message. From the construction of the code   
the $k-$th receiver $R_k$ obtains it's wanted message symbol from the code symbol $x_{k}+x_{k+m}+\dots+x_{k+(n-1)m}$ since all other message symbols appearing in the code symbol are its antidotes. Note that the receivers $R_{k+m}, R_{k+2m} , \cdots, R_{k+(n-1)m}$ also can use the same code symbol to obtain their respective wanted message symbols since all other messages symbols are their antidotes. Thus every receiver gets wanted message.   

The remaining part of the theorem can be shown in the same way as was done for Case I in Theorem \ref{thm1}.
\end{proof}

Note that By using the proposed code, all the receivers can decode their wanted messages by using just one transmitted code symbol. This means that for the proposed code every receiver uses the minimum possible code symbol. \\

\begin{example}
\label{ex2}
Let $D=16$ and  $K=20.$ The capacity for this case is $C=\frac{1}{K-D}=\frac{1}{4}.$ The code that achieves this capacity is\\
$\mathfrak{C}_{2}=\{x_{1}+x_{5}+x_{9}+x_{13}+x_{17}, \\
~~~~~~~~~~ x_{2}+x_{6}+x_{10}+x_{14}+x_{18}, \\
~~~~~~~~~~ x_{3}+x_{7}+x_{11}+x_{15}+x_{19}, \\
~~~~~~~~~~ x_{4}+x_{8}+x_{12}+x_{16}+x_{20}\}.
$ \\
The generator matrix  $L_{2}$ is the $20 \times 4$ matrix shown below. \\
\begin{center}
$L_{2} = \left[\begin{array}{cccc}
  1 & 0 & 0 & 0\\
  0 & 1 & 0 & 0\\
  0 & 0 & 1 & 0\\
  0 & 0 & 0 & 1\\
  1 & 0 & 0 & 0\\
  0 & 1 & 0 & 0\\
  0 & 0 & 1 & 0\\
  0 & 0 & 0 & 1\\
  1 & 0 & 0 & 0\\
  0 & 1 & 0 & 0\\
  0 & 0 & 1 & 0\\
  0 & 0 & 0 & 1\\
  1 & 0 & 0 & 0\\
  0 & 1 & 0 & 0\\
  0 & 0 & 1 & 0\\
  0 & 0 & 0 & 1\\
  1 & 0 & 0 & 0\\
  0 & 1 & 0 & 0\\
  0 & 0 & 1 & 0\\
  0 & 0 & 0 & 1\\
  \end{array}\right].$
\end{center}


\end{example}

~ \\

\noindent
{\bf Case III:} 
 $D-\frac{K}{2}$ divides $\frac{K}{2}$ and 
\begin{equation}
\label{antidote3}
{\cal K}_k=\{x_{k+\frac{K}{2}},x_{k+D-\frac{K}{2}},x_{k+D}\}.
\end{equation} 
\begin{theorem}
\label{thm3}
If $D-\frac{K}{2}$ divides $\frac{K}{2}$ and the antidote pattern is as in \eqref{antidote3} then the proposed scalar linear code is\\
$\mathfrak{C}=\{x_{i+jm}+x_{\frac{K}{2}+i+jm}+x_{i+(j+1)m}+x_{\frac{K}{2}+i+(j+1)m} \\
~~~~~~~~~~  ~| ~i =\{1,2,\dots,m\},\ j=\{0,1,2,\dots,n-2\}\}$ \\
where $D-\frac{K}{2}=m$ and $\frac{K/2}{D-\frac{K}{2}}=n$. The capacity of the proposed code is $C=\frac{1}{K-D}$ and the $minrank$ of the side information graph is $K-D$.
\end{theorem}
\begin{proof}
First we show that by using this code every receiver $R_{k}$ for $k \in \lceil K \rfloor$ can decode its wanted message.\\

\noindent
\emph{Case (i).} $k \in \{1,2,\dots,K-D\}$:\\
Let $k=i+jm$. As $i$ and $j$ run, $k$ runs from 1 to $K-D$. From the construction of the code, the code book consists of the code symbol of the form $x_{i+jm}+x_{\frac{K}{2}+i+jm}+x_{i+(j+1)m}+x_{\frac{K}{2}+i+(j+1)m}$ for $i = \{1,2,\dots,m\}$ and $j = \{0,1,2,\dots,n-2\}$. For the first message symbol of this code symbol the other three message symbols are antidotes. Thus the first message symbol $x_{i+jm} = x_{k}$ can be decoded by the receiver $R_{k}$.\\

\noindent
\emph{Case (ii).} $k \in \{\frac{K}{2}+1,\dots,\frac{3K}{2}-D \}$:\\
Let $k=\frac{K}{2}+i+jm$. As $i$ and $j$ run, $k$ runs from $\frac{K}{2}+1$ to $\frac{3K}{2}-D$. From the construction of the code, the code book consists of the code symbol of the form $x_{i+jm}+x_{\frac{K}{2}+i+jm}+x_{i+(j+1)m}+x_{\frac{K}{2}+i+(j+1)m}$ for $i = \{1,2,\dots,m\}$ and $j =\{0,1,2,\dots,n-2\}$. For the second message symbol of this code symbol the other three message symbols are antidotes. Thus the second message symbol $x_{\frac{K}{2}+i+jm} = x_{k}$ can be decoded by the receiver $R_{k}$.\\

\noindent
\emph{Case (iii).} $k \in \{K-D+1,K-D+2,\dots,\frac{K}{2}\}$:\\
Let $k=i+(n-1)m$. As $i$ and $j$ run, $k$ runs from $K-D+1$ to $\frac{K}{2}$. Let $\mathfrak{C}_{i,j}=x_{i+jm}+x_{\frac{K}{2}+i+jm}+x_{i+(j+1)m}+x_{\frac{K}{2}+i+(j+1)m}$ and $$\mathfrak{C}_{i}=\sum_{j=0}^{n-2} \mathfrak{C}_{i,j}$$ i.e. $\mathfrak{C}_{i}$ is obtained by adding the $n-1$ code symbols corresponding to given $i$ and for all $j$. Therefore, $\mathfrak{C}_{i}$=$x_{i}+x_{\frac{K}{2} + i}+x_{i+m}+x_{\frac{K}{2}+i+m} +x_{i+m}+x_{\frac{K}{2}+i+m} +x_{i+2m}+x_{\frac{K}{2}+i+2m}+\dots+x_{i+(n-3)m}+x_{\frac{K}{2}+i+(n-3)m}+ x_{i+(n-2)}+x_{\frac{K}{2} + i+(n-2)m}+ x_{i+(n-2)m}+x_{\frac{K}{2} + i+(n-2)m}+x_{i+(n-1)m}+x_{\frac{K}{2}+i+(n-1)m}$. In this sum, all the message symbols get canceled except the first two message symbols in $\mathfrak{C}_{i,0}$ and the last two message symbols in $\mathfrak{C}_{i,n-2}$ and the above sum is equal to $x_{i}+x_{\frac{K}{2} + i}+x_{i+(n-1)m}+x_{\frac{K}{2}+i+(n-1)m}$. In this sum for the receiver $R_{i+(n-1)m}$, other three message symbols are antidotes. Thus $x_{i+(n-1)m} = x_{k}$ can be decoded by the receiver $R_{k}$. The range of $k$ covered in this decoding process is the span of $i+(n-1)m$ and the span of $i+(n-1)m$ for $i = \{1,2,\dots,m\}$ is $\{K-D+1,K-D+2,\dots,K-D+m= \frac{K}{2}\}$.\\

\noindent
\emph{Case (iv).} $k \in \{\frac{3K}{2}-D+1,\dots,K\}$\\
Let $k=i+(n-1) m + \frac{K}{2}$. As $i$ and $j$ runs, $k$ runs from $\frac{3K}{2}-D+1$ to $K$. By adding the $n-1$ code symbols corresponding to given $i$ and for all $j$ as done in case III we get $\mathfrak{C_{i}}$=$x_{i}+x_{\frac{K}{2} + i}+x_{i+(n-1)m}+x_{\frac{K}{2}+i+(n-1)m}$. In this sum for the receiver $R_{i+(n-1)m+\frac{K}{2}}$, other three message symbols are antidotes. Thus $x_{i+(n-1)m+\frac{K}{2}} = x_{k}$ can be decoded by the receiver $R_{k}$. The range of $k$ covered in this decoding process is the span of  $i+(n-1)m+ \frac{K}{2}$ and the span of $i+(n-1)m+ \frac{K}{2}$ is $\{\frac{3K}{2}-D+1,\frac{3K}{2}-D+2,\dots\frac{3K}{2}-D+m= K\}$.\\
This completes the decoding process for all receivers.\\ 

The number of code symbols  transmitted is equal to the product of the number of values that $i$ takes and the number of values that $j$ takes. Number of code symbols = $m(n-1)$ = $(D$-$\frac{K}{2})$.($\frac{K/2}{D-\frac{K}{2}}-1)$=$K-D$. The remaining claims of the theorem can be shown along the same lines as in the  proof of Theorem \ref{thm1}.
\end{proof}

By using the proposed code \mbox{$K-2(D-\frac{K}{2})$} receivers can decode their wanted messages by using just one transmitted symbol. Remaining receivers decode their wanted messages by using $\frac{K/2}{D-\frac{K}{2}}-1$ transmissions each.\\

\begin{example}
\label{ex3}
Let $D=12,\ K=20.$ Then we have $C=\frac{1}{K-D}=\frac{1}{8}.$ The proposed code is \\
$\mathfrak{C_{3}}=\{x_{1}+x_{11}+x_{3}+x_{13}, ~~ x_{2}+x_{12}+x_{4}+x_{14}, \\
~~~~~~~~x_{3}+x_{13}+x_{5}+x_{15}, ~~ x_{4}+x_{14}+x_{6}+x_{16}, \\
~~~~~~~~x_{5}+x_{15}+x_{7}+x_{17}, ~~ x_{6}+x_{16}+x_{8}+x_{18}, \\
~~~~~~~~x_{7}+x_{17}+x_{9}+x_{19}, ~~ x_{8}+x_{18}+x_{10}+x_{20}\}.$
\noindent
The code is given by $\mathfrak{C_{3}}$=$\underline{x}L_{3}$, where $\underline{x}$=$[x_{1},x_{2}, \dots ,x_{20}]$ and $L_{3}$ is the  $20 \times 8$ matrix given below.
\begin{center}
$$ L_{3}= \left[\begin{array}{*{20}c}
   1 & 0 & 0 & 0 & 0 & 0 & 0 & 0 \\
   0 & 1 & 0 & 0 & 0 & 0 & 0 & 0 \\
   1 & 0 & 1 & 0 & 0 & 0 & 0 & 0 \\
   0 & 1 & 0 & 1 & 0 & 0 & 0 & 0 \\
   0 & 0 & 1 & 0 & 1 & 0 & 0 & 0 \\
   0 & 0 & 0 & 1 & 0 & 1 & 0 & 0 \\
   0 & 0 & 0 & 0 & 1 & 0 & 1 & 0 \\
   0 & 0 & 0 & 0 & 0 & 1 & 0 & 1 \\
   0 & 0 & 0 & 0 & 0 & 0 & 1 & 0 \\
   0 & 0 & 0 & 0 & 0 & 0 & 0 & 1 \\
   1 & 0 & 0 & 0 & 0 & 0 & 0 & 0 \\
   0 & 1 & 0 & 0 & 0 & 0 & 0 & 0 \\
   1 & 0 & 1 & 0 & 0 & 0 & 0 & 0 \\
   0 & 1 & 0 & 1 & 0 & 0 & 0 & 0 \\
   0 & 0 & 1 & 0 & 1 & 0 & 0 & 0 \\
   0 & 0 & 0 & 1 & 0 & 1 & 0 & 0 \\
   0 & 0 & 0 & 0 & 1 & 0 & 1 & 0 \\
   0 & 0 & 0 & 0 & 0 & 1 & 0 & 1 \\
   0 & 0 & 0 & 0 & 0 & 0 & 1 & 0 \\
   0 & 0 & 0 & 0 & 0 & 0 & 0 & 1 \\
  \end{array}\right].
$$
\end{center}
~ \\

The code in Theorem \ref{thm3} in general form can  be given by $\mathfrak{C}=\underline{x}L$ where $\underline{x}$=$[x_{1},x_{2}, \dots ,x_{K}]$ and $L$ is a $K \times K-D$ matrix where the columns of this matrix are the cyclic shifts of the column given below.\\
$R$ =$[\underbrace{\underbrace{1 0 0 \dots 0}_{\text{m elements}}\ \underbrace{1 0 0 \dots 0}_{\text{m elements}}\ \underbrace{0 0 0 \dots 0}_{\text{m elements}}\ \dots\ \underbrace{0 0 0 \dots 0}_{\text{m elements}}}_{\text{K/2 elements}}]$\\ \\
Column 1=$[R\ R]^{T}$. The remaining $K-D-1$ columns are the cyclic shifts of column 1. This is what has been  illustrated in Example \ref{ex3} above. \\

\end{example}
\noindent
{\bf Case IV:} 
$\frac{K}{2}-D$ divides $D$ and 
\begin{equation}
\label{antidote4}
{\cal K}_k=\{x_{k+\frac{K}{2}-D},x_{k+2(\frac{K}{2}-D)}+ \cdots +x_{k+D}\}.
\end{equation}
\begin{theorem}
\label{thm4}
For the case   $\frac{K}{2}-D$ divides $D$ with the antidote pattern given in \eqref{antidote4}, the proposed scalar linear code is,\\ 
$\mathfrak{C}=\{x_{i}+x_{i+m}+\dots+x_{i+pm}, \\ 
~~~~~~~~ x_{i+m}+x_{i+2m}+\dots+x_{i+(p+1)m}, \\
~~~~~~~~~~~~~~~~~~~~ \vdots \\
~~~~~~~~ x_{i+m(p+1)}+x_{i+m(p+2)}+\dots+x_{i+(n-1)m} \\
~~~~~~~~~~~~~~~~~~~~~~~ | i=\{1,2, \dots ,m\}\}$

\noindent
where  $m=\frac{K}{2}-D$, $n=\frac{K}{\frac{K}{2}-D}$ and $\frac{D}{\frac{K}{2}-D}=p$. \\
The capacity is $C=\frac{1}{K-D}$ and the $minrank$ of this side information graph is $K-D$.
\end{theorem}

\begin{proof}
By using this code every receiver can decode its wanted message as follows:\\

\noindent
\emph{Case (i).} $k \in \{1,2,...,K-D\}$:\\
Let $k=i+lm$. As $i$ and $l$ run, $k$ runs from 1 to $K-D$. From the construction of the code, the code book consists of the code symbol of the form  ${x_{i+l m}+x_{i+(l+1) m}+ \dots +x_{i+(l+q) m}}$ for $i = \{1,2,\dots ,m\}$ and $l = \{0,1,2,\dots,n-p-1\}$. For the receiver $R_{i+l m}$, in this code symbol all other message symbols are the antidotes. Thus the first message symbol $x_{i+l m} = x_{k}$ can be decoded by the receiver $R_{k}$. Therefore $m(n-p)=K-D$ message symbols can be decoded in this way by using one code symbol for each message symbol.\\

\noindent
\emph{Case (ii).} $k \in \{K-D+1,K-D+2,...,K\}$:\\
Let $k=i+(p+l) m$. As $i$ and $l$ run, $k$ runs from $K-D+1$ to $K$. From the construction of the code, the code book consists of the code symbol of the form  $C_{i}$ =$x_{i+(p+1)m}$+$x_{i+(p+2)m}$+ $\dots$ +$x_{i+(p+l-1) m}$+$x_{i+(p+l) m}$+$x_{i+(p+l+1) m}$+...+$x_{i+(n-1)m}$ for $i = \{1,2,\dots ,m\}$ and $l = \{2,3,\dots,n-p-1\}$. For a given $k=i+(p+l) m$, in the above code symbol $C_{i}$ the message symbols present before the required message symbol $x_{k}$ are in interference to $x_{k}$ and message symbols present after $x_{k}$ are  the antidotes of $R_{k}$. We shall cancel the interference by using the other code symbols. For the receiver $R_{i+(p+l) m}$, the message symbols $x_{i+(p+1) m},x_{i+(p+2) m},\dots x_{i+(p+l-1) m}$ are in interference and the message symbols $x_{i+(p+l+1) m}...x_{i+(n-1) m}$ are the antidotes. The interference can be canceled by adding the code symbols $x_{i}+x_{i+m}+\dots+x_{i+pm}$ and $x_{i+(l-1) m}$+$x_{i+l m}$+$x_{i+(l+1) m}$+ $\dots$ +$x_{i+(p+l-1) m}$. By adding the two code symbols we get $S_{1}$=$x_{i}+x_{i+m}+\dots+x_{i+(l-2)m}+x_{i+(p+1) m}$+$x_{i+(p+2) m}$+ $\dots$ +$x_{i+(p+l-1) m}$. By adding this sum $S_{1}$ to the code symbol $C_{i}$ we get $S_{2}$= $x_{i}+x_{i+m}+\dots+x_{i+(l-2)m}+ x_{i+(p+l) m}$+$x_{i+(p+l+1) m}$+...+$x_{i+(n-1) m}$. For $R_{i+(p+l) m}$, all other message symbols in $S_{2}$ are antidotes and thus $x_{i+(p+l) m}$ can be decoded from $S_{2}$.\\
For example, if $k=i+(p+2) m$, in the code symbol $x_{i+m  (p+1)}+x_{i+m  (p+2)}+\dots+x_{i+(n-1)  m}$, the message symbol $x_{i+m(p+1)}$ is the interference to receiver $R_{i+m(p+2)}$ and all other message symbols in the code symbol are antidotes. The interference message symbol $x_{i+m  (p+1)}$ is required to be replaced with antidote. By adding  $x_{i}+x_{i+m}+\dots+x_{i+p  m}$ and $x_{i+m}+x_{i+2  m}+\dots+x_{i+(p+1)  m}$, we get $x_{i}+x_{i+(p+1)  m}$. By adding $x_{i}+x_{i+(p+1)  m}$ and $x_{i+m  (p+1)}+x_{i+m  (p+2)}+\dots+x_{i+(n-1)  m}$ we get $x_{i}+x_{i+m (p+2)}+\dots+x_{i+(n-1)  m}$ and we can decode $x_{i+m  (p+2)}$ from this. \\
The number of code symbols transmitted is equal to $m  (p+2)=(\frac{K}{2}-D)  (\frac{D}{\frac{K}{2}-D}+2)={K-D}$. Hence the Capacity = $\dfrac{1}{K-D}$ and the minrank of the corresponding side information graph is $K-D.$
\end{proof}
By using the proposed code, \mbox{$K-D$} receivers can decode their wanted messages by using just one transmitted symbol. Remaining receivers decode their wanted messages by using 3 transmissions each.\\

The general form of the $L$ matrix for the code in Theorem \ref{thm4}  is a $K \times K-D$ matrix where the columns of this matrix are the cyclic shifts of the column given below.\\
$R_{1}$ =$[\underbrace{\underbrace{1 0 0 \dots 0}_{\text{m elements}}\ \underbrace{1 0 0 \dots 0}_{\text{m elements}}\ \dots\ \underbrace{1 0 0 \dots 0}_{\text{m elements}}}_{\text{(p+1)m elements}}]$\\ \\
$R_{2}$ =$[\underbrace{\underbrace{0 0 0 \dots 0}_{\text{m elements}}\ \underbrace{0 0 0 \dots 0}_{\text{m elements}}\ \dots\ \underbrace{0 0 0 \dots 0}_{\text{m elements}}}_{\text{(n-p-1)m elements}}]$\\ \\
Column 1=$[R_{1}\ R_{2}]^{T}$. The remaining $K-D-1$ columns are the cyclic shifts of column 1. This is illustrated in the following example.
\begin{example}
\label{ex4}
Let $D=8, K=20$ and $C=\frac{1}{K-D}=\frac{1}{12}.$ The proposed code is\\
$\mathfrak{C_{4}}=\{x_{1}+x_{3}+x_{5}+x_{7}+x_{9}, ~~ x_{2}+x_{4}+x_{6}+x_{8}+x_{10}, \\
x_{3}+x_{5}+x_{7}+x_{9}+x_{11}, ~~ x_{4}+x_{6}+x_{8}+x_{10}+x_{12},\\ 
x_{5}+x_{7}+x_{9}+x_{11}+x_{13}, ~~ x_{6}+x_{8}+x_{10}+x_{12}+x_{14},\\ 
x_{7}+x_{9}+x_{11}+x_{13}+x_{15}, ~~ x_{8}+x_{10}+x_{12}+x_{14}+x_{16},\\
x_{9}+x_{11}+x_{13}+x_{15}+x_{17}, ~~ x_{10}+x_{12}+x_{14}+x_{16}+x_{18},\\
x_{11}+x_{13}+x_{15}+x_{17}+x_{19}, ~~ x_{12}+x_{14}+x_{16}+x_{18}+x_{20}\}. \\
$
\noindent
The generator matrix for this code is the $20 \times 12$ matrix  $L_{4}$  shown below.

\noindent
{\bf Case V:}
There is an integer $\lambda$ such that $D$ divides $K-\lambda$ and $\lambda$ divides $D$  and
\begin{equation}
\label{antidote5}
{\cal K}_k=\left\{
                \begin{array}{ll}
                  x_{k+D},$ if $\ k\leq K-D-\lambda\\
                  \{x_{k},x_{k+\lambda},\dots,x_{k+D}\},$ if $K-D-\lambda<k\leq K
                  \end{array}
              \right.
\end{equation}
~ \\
\begin{center}
$$ L_{4}= \left[\begin{array}{*{20}c}
   1 & 0 & 0 & 0 & 0 & 0 & 0 & 0 & 0 & 0 & 0 & 0 \\
   0 & 1 & 0 & 0 & 0 & 0 & 0 & 0 & 0 & 0 & 0 & 0 \\
   1 & 0 & 1 & 0 & 0 & 0 & 0 & 0 & 0 & 0 & 0 & 0 \\
   0 & 1 & 0 & 1 & 0 & 0 & 0 & 0 & 0 & 0 & 0 & 0 \\
   1 & 0 & 1 & 0 & 1 & 0 & 0 & 0 & 0 & 0 & 0 & 0 \\
   0 & 1 & 0 & 1 & 0 & 1 & 0 & 0 & 0 & 0 & 0 & 0 \\
   1 & 0 & 1 & 0 & 1 & 0 & 1 & 0 & 0 & 0 & 0 & 0 \\
   0 & 1 & 0 & 1 & 0 & 1 & 0 & 1 & 0 & 0 & 0 & 0 \\
   1 & 0 & 1 & 0 & 1 & 0 & 1 & 0 & 1 & 0 & 0 & 0 \\
   0 & 1 & 0 & 1 & 0 & 1 & 0 & 1 & 0 & 1 & 0 & 0 \\
   0 & 0 & 1 & 0 & 1 & 0 & 1 & 0 & 1 & 0 & 1 & 0 \\
   0 & 0 & 0 & 1 & 0 & 1 & 0 & 1 & 0 & 1 & 0 & 1 \\
   0 & 0 & 0 & 0 & 1 & 0 & 1 & 0 & 1 & 0 & 1 & 0 \\
   0 & 0 & 0 & 0 & 0 & 1 & 0 & 1 & 0 & 1 & 0 & 1 \\
   0 & 0 & 0 & 0 & 0 & 0 & 1 & 0 & 1 & 0 & 1 & 0 \\
   0 & 0 & 0 & 0 & 0 & 0 & 0 & 1 & 0 & 1 & 0 & 1 \\
   0 & 0 & 0 & 0 & 0 & 0 & 0 & 0 & 1 & 0 & 1 & 0 \\
   0 & 0 & 0 & 0 & 0 & 0 & 0 & 0 & 0 & 1 & 0 & 1 \\
   0 & 0 & 0 & 0 & 0 & 0 & 0 & 0 & 0 & 0 & 1 & 0 \\
   0 & 0 & 0 & 0 & 0 & 0 & 0 & 0 & 0 & 0 & 0 & 1 \\
   \end{array}\right]$$
\end{center}

\noindent
\end{example}

\begin{theorem}
\label{thm5}
If  $D$ divides $K-\lambda$ and $\lambda$ divides $D$ and the antidote pattern is as given by \eqref{antidote5} then the scalar linear code is given by \\
$\mathfrak{C}=\{{x_{i+(j-1)D}+x_{i+jD}}|\ i = \{1,2,\dots,D\},\ j= \{1,2,\dots,n-1\}\}\\ \cup \{{x_{K-\lambda+r}+x_{K-\lambda+r-\lambda}+\dots+x_{K-\lambda+r-t\lambda}} \\
~~~~~~~~~~~|\ r = \{1,2,\dots,\lambda\}, \ t = \{1,2,\dots,\frac{D}{\lambda}$\}\} \\
 for $\frac{K-\lambda}{D}>1$  and $\frac{K-\lambda}{D}=n$.\\
The capacity is $C=\frac{1}{K-D}$ and $minrank$ of the corresponding side information graph is $K-D$.
\end{theorem}

\begin{proof}
Every receiver can decode its wanted message as follows:\\

\noindent
\emph{Case (i).} $k\leq {K-D-\lambda}$:\\
The $k^{th}$  receiver wants to decode $x_{k}$. The code book consists of the code symbol of the form $x_{i+(j-1)D}+x_{i+jD}$. As $i$ and $j$ runs, $x_{i+(j-1)D}+x_{i+jD}$ runs from $x_{1}+x_{1+D}$ to $x_{K-D-\lambda}+x_{K-\lambda}$. Therefore for every {$k\leq {K-D-\lambda}$}, the code book consists of the symbol $x_{k} + x_{k+D}$. Since the symbol $x_{k+D}$ is the antidote of $k^{th}$ receiver, $k^{th}$ receiver can decode its required symbol $x_{k}$. \\

\noindent
\emph{Case (ii).} ${K-D-\lambda}<k \leq{K-D}$:\\
To decode $x_{k}$ in this range, let \mbox{$k=K-D-\lambda+r$}. From the construction of the code, the code book consists of the code symbol $x_{K-\lambda+r-\frac{D}{\lambda}\lambda}+x_{K-\lambda+r-(\frac{D}{\lambda}-1)\lambda}+\dots+x_{K-\lambda+r-\lambda}+x_{K-\lambda+r}$ where $r = \{1,2,\dots,\lambda\}$. For the first message symbol $x_{K-\lambda+r-\frac{D}{\lambda}\lambda}=x_{k}$ in these code symbols every other message symbol is the antidote. Thus we can decode $x_{k}$ in the range ${K-D-\lambda}<k \leq{K-D}$.\\

\noindent
\emph{Case (iii).} $k>K-D$:\\
To decode $x_{k}$ in this range, let $k=K-\lambda+r-s \lambda$ where \mbox{$r= \{1,2,\dots,\lambda\}$} and \mbox{$s = \{0,1,2,\dots,\frac{D}{\lambda}-1\}$}. The code book consists of the code symbol  $C_{r}=x_{K-\lambda+r-\frac{D}{\lambda}\lambda}+x_{K-\lambda+r-(\frac{D}{\lambda}-1)\lambda}+\dots+x_{K-\lambda+r-(s+1) \lambda}+\mathbf{x_{K-\lambda+r-s \lambda}} +x_{K-\lambda+r+(s-1) \lambda}+\dots+x_{K-\lambda+r-\lambda}+x_{K-\lambda+r}$. For the receiver $R_{k}=R_{K-\lambda+r-s \lambda}$, the message symbols present before $x_{k}$ in the above code symbol $C_{r}$ are interference and the message symbols present after $x_{k}$ are the antidotes. For every ${K-D-\lambda}<j \leq {K-\lambda}$, the code book consists of the code symbols $x_{j}+x_{j-D},\ x_{j-D}+x_{j-2D},\dots,x_{(j\ mod\ D)+D}+x_{j\ mod\ D}$. By adding these $\frac{K-\lambda}{D}-1$ code symbols we get $S_{j}$ = $x_{j}+x_{j\ mod\ D}$.  For the $k^{th}$ receiver, the antidotes are $x_{(k+1)\ mod\ K}, x_{(k+2)\ mod\ K},\dots,x_{(k+D)\ mod\ K}$. By using the symbols $S_{j}$ of the form $x_{j}+x_{j\ mod\ D}$ for ${K-D-\lambda}<j \leq {K-\lambda}$, the interference in the code symbol $C_{r}$ given above can be canceled for the required message symbol $x_{k}$. The symbol $x_{K-\lambda-r-(s+1) \lambda}=x_{k-\lambda}$ is interference to $R_{k}$. By summing the code symbols $x_{j}+x_{j-D},\ x_{j-D}+x_{j-2D},\dots,x_{(j\ mod\ D)+D}+x_{j\ mod\ D}$ where $j= K-\lambda-r-(s+1) \lambda$, we get $x_{K-\lambda-r-(s+1) \lambda} +x_{D-r-(s+1) \lambda}$. The message symbol $x_{D-r-(s+1) \lambda}$ is in the antidote to $R_{k}=R_{K-\lambda+r-s \frac{D}{\lambda}}$.  Thus every symbol for $k>K-D$ can be decoded.\\

The number of code symbols transmitted is equal to $\lambda$ more than the product of the number of values that $i$ takes and the number of values that $j$ takes. Number of code symbols=$D(n$-$1)$+$\lambda$=$D$.$(\frac{K-\lambda}{D}$-$1)$+$\lambda$=$K$-$D$. Now the other claims of the theorem follows along the same lines as in the proof of Theorem \ref{thm1}.



\end{proof}

By using the proposed code, \mbox{$K-D$} receivers can decode their wanted messages by using just one transmitted symbol. 

\begin{example}
\label{ex5}
$D=4,\ K=21,\ \lambda=1.$  $C=\frac{1}{K-D}=\frac{1}{17}.$ \\ 
Code $\mathfrak{C_{5}}=\{x_{1}+x_{5}, x_{2}+x_{6}, x_{3}+x_{7},x_{4}+x_{8}, x_{5}+x_{9}, \\
~~~~~~~~ x_{6}+x_{10}, x_{7}+x_{11}, x_{8}+x_{12}, x_{9}+x_{13}, x_{10}+x_{14}, \\
~~~~~~~~ x_{11}+x_{15}, x_{12}+x_{16}, x_{13}+x_{17}, x_{14}+x_{18}, x_{15}+x_{19}, \\
~~~~~~~~x_{16}+x_{20}, x_{17}+x_{18}+x_{19}+x_{20}+x_{21}\}$. \\
In terms of $L$ matrix this code is given by the  $21 \times 17$ matrix $L_5$ given below.

\begin{center}
{\tiny
$$ L_{5}= \left[\begin{array}{*{20}c}
   1 & 0 & 0 & 0 & 0 & 0 & 0 & 0 & 0 & 0 & 0 & 0 & 0 & 0 & 0 & 0 & 0 \\
   0 & 1 & 0 & 0 & 0 & 0 & 0 & 0 & 0 & 0 & 0 & 0 & 0 & 0 & 0 & 0 & 0 \\
   0 & 0 & 1 & 0 & 0 & 0 & 0 & 0 & 0 & 0 & 0 & 0 & 0 & 0 & 0 & 0 & 0 \\
   0 & 0 & 0 & 1 & 0 & 0 & 0 & 0 & 0 & 0 & 0 & 0 & 0 & 0 & 0 & 0 & 0 \\
   1 & 0 & 0 & 0 & 1 & 0 & 0 & 0 & 0 & 0 & 0 & 0 & 0 & 0 & 0 & 0 & 0 \\
   0 & 1 & 0 & 0 & 0 & 1 & 0 & 0 & 0 & 0 & 0 & 0 & 0 & 0 & 0 & 0 & 0 \\
   0 & 0 & 1 & 0 & 0 & 0 & 1 & 0 & 0 & 0 & 0 & 0 & 0 & 0 & 0 & 0 & 0 \\
   0 & 0 & 0 & 1 & 0 & 0 & 0 & 1 & 0 & 0 & 0 & 0 & 0 & 0 & 0 & 0 & 0 \\
   0 & 0 & 0 & 0 & 1 & 0 & 0 & 0 & 1 & 0 & 0 & 0 & 0 & 0 & 0 & 0 & 0 \\
   0 & 0 & 0 & 0 & 0 & 1 & 0 & 0 & 0 & 1 & 0 & 0 & 0 & 0 & 0 & 0 & 0 \\
   0 & 0 & 0 & 0 & 0 & 0 & 1 & 0 & 0 & 0 & 1 & 0 & 0 & 0 & 0 & 0 & 0 \\
   0 & 0 & 0 & 0 & 0 & 0 & 0 & 1 & 0 & 0 & 0 & 1 & 0 & 0 & 0 & 0 & 0 \\
   0 & 0 & 0 & 0 & 0 & 0 & 0 & 0 & 1 & 0 & 0 & 0 & 1 & 0 & 0 & 0 & 0 \\
   0 & 0 & 0 & 0 & 0 & 0 & 0 & 0 & 0 & 1 & 0 & 0 & 0 & 1 & 0 & 0 & 0 \\
   0 & 0 & 0 & 0 & 0 & 0 & 0 & 0 & 0 & 0 & 1 & 0 & 0 & 0 & 1 & 0 & 0 \\
   0 & 0 & 0 & 0 & 0 & 0 & 0 & 0 & 0 & 0 & 0 & 1 & 0 & 0 & 0 & 1 & 0 \\
   0 & 0 & 0 & 0 & 0 & 0 & 0 & 0 & 0 & 0 & 0 & 0 & 1 & 0 & 0 & 0 & 1 \\
   n0 & 0 & 0 & 0 & 0 & 0 & 0 & 0 & 0 & 0 & 0 & 0 & 0 & 1 & 0 & 0 & 1 \\
   0 & 0 & 0 & 0 & 0 & 0 & 0 & 0 & 0 & 0 & 0 & 0 & 0 & 0 & 1 & 0 & 1 \\
   0 & 0 & 0 & 0 & 0 & 0 & 0 & 0 & 0 & 0 & 0 & 0 & 0 & 0 & 0 & 1 & 1 \\
   0 & 0 & 0 & 0 & 0 & 0 & 0 & 0 & 0 & 0 & 0 & 0 & 0 & 0 & 0 & 0 & 1 \\
   \end{array}\right]$$
}
\end{center}

\end{example}
\noindent
{\bf Case VI:} 
There is an integer $\lambda$ such that $K-D$ divides $K-\lambda,$  $\lambda$ divides $K-D$ and the antidote pattern is given by \eqref{antidote1}.
\begin{theorem}
\label{thm6}
If $K-D$ divides $K-\lambda,$  $\lambda$ divides $(K-D)$ and the antidote pattern is as in \eqref{antidote1} then  the proposed scalar linear code is \\
$\mathfrak{C}=\{x_{i}+x_{i+m}+\dots+x_{i+(q-1)m}+x_{qm+1+(i-1) mod \lambda) }\} \\ 
~~~~~~~~~~ |\ i = \{1,2,\dots,m\}\}$ \\
where $K-D=m$, and $\frac{K-\lambda}{K-D}=q$. The proposed code is optimal.
\end{theorem}

\begin{proof}
The number of code symbols transmitted is equal to the number of values that $i$ takes, which is equal to $K-D$.  In \cite{MCJ} the capacity of this problem is shown to be $\frac{1}{K-D}.$ This proves the optimality, when every receiver can decode its wanted message which is shown below: \\

\noindent
\emph{Case (i).} $k\in \{1,2,...,m-\lambda\}$:\\
 The $k^{th}$  receiver wants to decode $x_{k}$.  From the construction of the code for every \mbox{$k\in \{1,2,...,m-\lambda\}$}, the code book consists of the code symbols $\{x_{k}+x_{k+m}+\dots+x_{k+(q-1)m}+x_{qm+1+(k-1) mod \lambda)}\}$ and $\{x_{k+\lambda}+x_{k+\lambda+m}+\dots+x_{k+\lambda+(q-1)m}+x_{qm+1+(k-1) mod \lambda)}\}$. By adding the code symbols, the last message symbol in both the code words gets canceled and we get $\{x_{k}+x_{k+m}+\dots+x_{k+(q-1)m}+x_{k+\lambda}+x_{k+\lambda+m}+\dots+x_{k+\lambda+(q-1)m}\}$ . Here every other message symbol is in antidote of $R_{k}$. Thus $k^{th}$ receiver where $k \in \{1,2,\dots,m-\lambda\}$ can decode its message.\\
~
\noindent
\emph{Case (ii).} $k\in \{m-\lambda+1,m-\lambda+2,...,m\}$:\\
 The message symbol $x_{k}$ can be decoded from the $k^{th}$ code symbol because all other message symbols in this code symbols are in the antidotes to $R_{k}$.\\

\noindent
\emph{Case (iii).} $k\in \{m+1,m+2,...,mq\}$:\\
Let \mbox{$l=1+(k-1)$ mod $m$}. The code book consists of the code symbol $C_{l}=\{x_{l}+x_{l+m}+\dots+x_{l+(q-1)m}+x_{qm+1+(l-1) mod \lambda) }\}$. The message symbol $x_{k}$ can be decoded from the $l$ code symbol because all other message symbol in the $l^{th}$ code symbol are antidotes to $R_{k}$. \\

\noindent
\emph{Case (iv).} $k\in \{K-\lambda+1,K-\lambda+2,...,K\}$:\\
Let \mbox{$l=k$ mod $(K-D)$}. The message symbol $x_{k}$ can be decoded from the code symbol $C_{l}$=$\{x_{l}+x_{l+m}+\dots+x_{l+(q-1)m}+x_{qm+1+(l-1) mod \lambda) }\}$ because the message symbol $x_{k}$ symbol is present in $C_{l}$ and all other message symbols are in antidotes to $R_{k}$.



\end{proof}

By using the proposed code, \mbox{$D+\lambda$} receivers can decode their wanted messages by using just one transmitted symbol. All other remaining receivers can decode the wanted message by using two transmissions.

\begin{example}
\label{ex6}
Let $D=17, K=21, \lambda=1.$ Then we have $C=\frac{1}{K-D}=\frac{1}{4}.$ \\
Code $\mathfrak{C_{6}}=\{x_{1}+x_{5}+x_{9}+x_{13}+x_{17}+x_{21}, \\
~~~~~~~~~~~~~~~~~ x_{2}+x_{6}+x_{10}+x_{14}+x_{18}+x_{21}, \\
~~~~~~~~~~~~~~~~~ x_{3}+x_{7}+x_{11}+x_{15}+x_{19}+x_{21}, \\
~~~~~~~~~~~~~~~~~ x_{4}+x_{8}+x_{12}+x_{16}+x_{20}+x_{21}\}.$  

The following $21 \times 4$ matrix $L_6$ is the generator matrix for this code. \\
\begin{center}
{\footnotesize 
$$ L_{6}= \left[\begin{array}{*{20}c}
   1 & 0 & 0 & 0 \\
   0 & 1 & 0 & 0 \\
   0 & 0 & 1 & 0 \\
   0 & 0 & 0 & 1 \\
   1 & 0 & 0 & 0 \\
   0 & 1 & 0 & 0 \\
   0 & 0 & 1 & 0 \\
   0 & 0 & 0 & 1 \\
   1 & 0 & 0 & 0 \\
   0 & 1 & 0 & 0 \\
   0 & 0 & 1 & 0 \\
   0 & 0 & 0 & 1 \\
   1 & 0 & 0 & 0 \\
   0 & 1 & 0 & 0 \\
   0 & 0 & 1 & 0 \\
   0 & 0 & 0 & 1 \\
   1 & 0 & 0 & 0 \\
   0 & 1 & 0 & 0 \\
   0 & 0 & 1 & 0 \\
   0 & 0 & 0 & 1 \\
   1 & 1 & 1 & 1 \\
   \end{array}\right]$$
}
\end{center}

\end{example}
\noindent
{\bf Case VII:} 
There is an integer $\lambda$ such that $D+\lambda$ divides $K$ and $\lambda$ divides $D$ and 
\begin{equation}
\label{antidote7}
{\cal K}_k=\{x_{k+\lambda},x_{k+2\lambda},...,x_{k+D}\}.
\end{equation}
\begin{theorem}
\label{thm7}
For the case  $D+\lambda$ divides $K,$  $\lambda$ divides $D$ and the antidote pattern is as in \eqref{antidote7}, the proposed code is  \\
$\mathfrak{C}=\{x_{i+j\lambda}+x_{i+(j+1)\lambda}+\dots+x_{i+(j+p)\lambda}\\
~~~~~~~~~~~~~~~  | i =\{1,2,\dots,\lambda\},\ j= \{1,2,\dots,\frac{K-D-\lambda}{\lambda}$\}\} \\
where  $\frac{D}{\lambda}=p$ and $\frac{K}{D+\lambda}=n.$ \\
The capacity is $C=\frac{1}{K-D}$ and $minrank$ of the side information graph of this antidote pattern is $K-D$.
\end{theorem}
\begin{proof}
That the proposed code satisfies the requirement of all the receivers is shown below:\\

\noindent
\emph{Case (i).} $k \leq K-D$:\\
Let \mbox{$k=(i+j\lambda)$}. The $k^{th}$  receiver wants to decode $x_{k}$. As $i$ and $j$ runs, $k$ runs from 1 to $K-D$. From the construction of the code for every $k=(i+j\lambda)\leq(K-D)$, the code book consists of the code symbol $x_{i+j\lambda}+x_{i+(j+1)\lambda}+\dots+x_{i+(j+p)\lambda}$ for \mbox{$i = \{1,2,\dots,\lambda\}$}. Since the message symbols $x_{k+\lambda}$, $x_{k+2\lambda}$, \dots , $x_{k+p\lambda}$ are in the antidote to $R_{k} = R_{i+j\lambda}$, $k^{th}$ receiver can decode its required message symbol $x_{k}$. \\

\noindent
\emph{Case (ii).} $k > K-D$:\\
Let $k=i+K-D+l\lambda$ for $i = \{1,2,\dots,\lambda\}$ and $l = \{0,1,2,\dots,p-1\}$. As $i$ and $l$ runs, $k$ runs from $K-D+1$ to $K$. For $k > K-D$, the code book consists of the code symbol $x_{i+K-D-\lambda}+x_{i+K-D}+\dots+x_{i+K-D+(l-1)\lambda}+\underline{x_{i+K-D+l\lambda}}+x_{i+K-D+(l+1)\lambda}+\dots+x_{i+(K-\lambda)}$ where $i = \{1,2,\dots,\lambda\}$ such that $x_{k}=x_{i+K-D+l\lambda}$ is present in this code symbol. In the above code symbol every message symbol before $x_{k}$ is in interference to $R_{k}$ and every message symbol after $x_{k}$ is the antidote to $R_{k}$. We will cancel the interference by using other code symbols.\\

Consider $C_{i,j}=\{x_{i+j\lambda}+x_{i+(j+1)\lambda}+\dots+x_{i+(j+p)\lambda}\}$, then we can write $C_{i,q}=\{x_{i+K-D-\lambda}+x_{i+K-D}+\dots+x_{i+K-D+(l-1)\lambda}+\underline{x_{i+K-D+l\lambda}}+x_{i+K-D+(l+1)\lambda}+\dots+x_{i+(K-\lambda)}\}$.\\
From the sum $$S=\sum_{s=0}^{n-2} \{C_{i,s(p+1)} + C_{i,(l+1)+s(p+1)}\}$$ we get $S=x_{i+K-D-\lambda}+x_{i+K-D}+\dots+x_{i+K-D+(l-1)\lambda}+ x_{i}+x_{i +\lambda}+\dots+x_{i +l\lambda}$. In $S$, the message symbols $x_{i+K-D-\lambda},x_{i+K-D},\dots,x_{i+K-D+(l-1)\lambda}$ are interference to $R_{k}$ and the message symbols $x_{i},x_{i +\lambda},\dots,x_{i +l\lambda}$ are antidotes to $R_{k}$. Thus by computing $S$ we can establish the relation between the interference and antidotes of $R_{i}$. By adding $S$ to $C_{i,q}$, we get $\{x_{i+K-D-\lambda}+x_{i+K-D}+\dots+x_{i+K-D+(l-1)\lambda}+ x_{i}+x_{i +\lambda}+\dots+x_{i +l\lambda}\}$+$\{x_{i+K-D-\lambda}+x_{i+K-D}+\dots+x_{i+K-D+(l-1)\lambda}+x_{i+K-D+l\lambda}+x_{i+K-D+(l+1)\lambda}+\dots+x_{i+(K-\lambda)}\}$=$x_{i}+x_{i +\lambda}+\dots+x_{i +l\lambda}+x_{i+K-D+l\lambda}+x_{i+K-D+(l+1)\lambda}+\dots+x_{i+(K-\lambda)}$. In this sum every message symbol is a antidote to $R_{k}$, thus $x_{k}$ can be decoded.

This completes the proof for the claim that the code satisfies the requirements of all the receivers. \\

The rest of the claims can be proved along the same lines as in the proof of Theorem \ref{thm1} noting that the number of code symbols transmitted is equal to the product of the number of values that $i$ takes and the number of values that $j$ takes. Number of code symbols=$\lambda(\frac{K-D}{\lambda})$ =$K-D$.\\
\end{proof}

By using the proposed code, \mbox{$K-D$} receivers can decode their wanted messages by using just one transmitted symbol. \\
\begin{example}
\label{ex7}
Let $U=0,\ D=5,\ K=18,\ \lambda=1.$ Then $C=\frac{1}{K-D}=\frac{1}{13}.$ 
Theorem \ref{thm7} gives the code \\
$\mathfrak{C}_7= \{
x_{1}+x_{2}+x_{3}+x_{4}+x_{5}+x_{6},~ x_{2}+x_{3}+x_{4}+x_{5}+x_{6}+x_{7}, \\
x_{3}+x_{4}+x_{5}+x_{6}+x_{7}+x_{8}, ~ x_{4}+x_{5}+x_{6}+x_{7}+x_{8}+x_{9}, \\
x_{5}+x_{6}+x_{7}+x_{8}+x_{9}+x_{10}, ~ x_{6}+x_{7}+x_{8}+x_{9}+x_{10}+x_{11}, \\
x_{7}+x_{8}+x_{9}+x_{10}+x_{11}+x_{12}, ~ x_{8}+x_{9}+x_{10}+x_{11}+x_{12}+x_{13}, \\
x_{9}+x_{10}+x_{11}+x_{12}+x_{13}+x_{14}, ~ x_{10}+x_{11}+x_{12}+x_{13}+x_{14}+x_{15}, \\
x_{11}+x_{12}+x_{13}+x_{14}+x_{15}+x_{16}, ~ x_{12}+x_{13}+x_{14}+x_{15}+x_{16}+x_{17}, \\
x_{13}+x_{14}+x_{15}+x_{16}+x_{17}+x_{18}\}$. \\

The generator matrix for this code is given by 
\begin{center}
{\tiny 
$ L_{7}= \left[\begin{array}{*{20}c}
   1 & 0 & 0 & 0 & 0 & 0 & 0 & 0 & 0 & 0 & 0 & 0 & 0 \\
   1 & 1 & 0 & 0 & 0 & 0 & 0 & 0 & 0 & 0 & 0 & 0 & 0 \\
   1 & 1 & 1 & 0 & 0 & 0 & 0 & 0 & 0 & 0 & 0 & 0 & 0 \\
   1 & 1 & 1 & 1 & 0 & 0 & 0 & 0 & 0 & 0 & 0 & 0 & 0 \\
   1 & 1 & 1 & 1 & 1 & 0 & 0 & 0 & 0 & 0 & 0 & 0 & 0 \\
   1 & 1 & 1 & 1 & 1 & 1 & 0 & 0 & 0 & 0 & 0 & 0 & 0 \\
   0 & 1 & 1 & 1 & 1 & 1 & 1 & 0 & 0 & 0 & 0 & 0 & 0 \\
   0 & 0 & 1 & 1 & 1 & 1 & 1 & 1 & 0 & 0 & 0 & 0 & 0 \\
   0 & 0 & 0 & 1 & 1 & 1 & 1 & 1 & 1 & 0 & 0 & 0 & 0 \\
   0 & 0 & 0 & 0 & 1 & 1 & 1 & 1 & 1 & 1 & 0 & 0 & 0 \\
   0 & 0 & 0 & 0 & 0 & 1 & 1 & 1 & 1 & 1 & 1 & 0 & 0 \\
   0 & 0 & 0 & 0 & 0 & 0 & 1 & 1 & 1 & 1 & 1 & 1 & 0 \\
   0 & 0 & 0 & 0 & 0 & 0 & 0 & 1 & 1 & 1 & 1 & 1 & 1 \\
   0 & 0 & 0 & 0 & 0 & 0 & 0 & 0 & 1 & 1 & 1 & 1 & 1 \\
   0 & 0 & 0 & 0 & 0 & 0 & 0 & 0 & 0 & 1 & 1 & 1 & 1 \\
   0 & 0 & 0 & 0 & 0 & 0 & 0 & 0 & 0 & 0 & 1 & 1 & 1 \\
   0 & 0 & 0 & 0 & 0 & 0 & 0 & 0 & 0 & 0 & 0 & 1 & 1 \\
   0 & 0 & 0 & 0 & 0 & 0 & 0 & 0 & 0 & 0 & 0 & 0 & 1 \\
        
  \end{array}\right]$
}
\end{center}
\end{example}

~ \\

\noindent
{\bf Case VIII:} 
There is an integer $\lambda$ such that $K-D+\lambda$ divides $K$ and $\lambda$ divides $k-D$ and 
\begin{equation}
\label{antidote8}
{\cal K}_k=\left\{
                \begin{array}{ll}
                  \{x_{k+\lambda},\ x_{k+\lambda+(K-D)},\\x_{k+2\lambda+(K-D)},\ x_{k+2\lambda+2(K-D)},\\ ~~~~~~~~~\vdots ~~~~~~~~~~~~~~~~~~~ \vdots\\x_{k+(p-1)\lambda+(p-2)(K-D)},\ x_{k+(p-1)\lambda+(p-1)(K-D)},\\x_{k+p\lambda+(p-1)(K-D)}\}\\
                  \end{array}
              \right.
\end{equation}
where $p=\frac{K}{K-D+\lambda}$.\\

\begin{theorem}
\label{thm8}
If \mbox{$K-D+\lambda$} divides $K$ and $\lambda$ divides \mbox{$K-D$} and the antidote pattern is as in \eqref{antidote8}, then the scalar linear code is given by\\ 
$
\mathfrak{C}=\{x_{i}+x_{i+\lambda}+x_{i+\lambda+(K-D)}+x_{i+2\lambda+(K-D)}\\
~~~~ +x_{i+2\lambda+2(K-D)}+x_{i+3\lambda+2(K-D)} \\
~~~~ +\dots+x_{i+(p-1)\lambda+(p-1)(K-D)}+x_{i+p\lambda+(p-1)(K-D)}\\
~~~~~ ~|~i= 1,2,\dots,K-D\}
$\\ 
where $\frac{K}{K-D+\lambda}=p$ and $\frac{K-D}{\lambda}=m$ is of optimal length.
\end{theorem}

\begin{proof}
We first prove that every receiver can decode its wanted message from the proposed  code.\\

\begin{scriptsize}
\noindent
Let $\mathcal{A}_{1}$=\mbox{$ \{1,2,\dots,K-D\}$}$\cup$\\ \{$K-D+\lambda+1,\dots,2(K-D)+\lambda\}\cup$\\\{$2(K-D)+2\lambda+1,\dots,3(K-D)+2\lambda\}\cup$\\
$~~~~~~~~~~~~~~~~~~~~~~~~~~~~~~ \vdots\\ \{(p-1)(K-D)+(p-1)\lambda+1,\ p(K-D)+(p-1)\lambda\}$ and \\ $\mathcal{A}_{2}$=$\{1,2,\dots,K\}\ \backslash\ \mathcal{A}_{1}$.\\
\end{scriptsize}

\noindent
\emph{Case (i).} $k \in \mathcal{A}_{1}$: The $k^{th}$  receiver wants to decode $x_{k}$. Let \mbox{$i=k$ mod $(K-D+\lambda)$}. From the construction of the code for every $k$ in the given range the code book consists of the code symbol $C_{i}$ = $x_{i}+x_{i+\lambda}+x_{i+\lambda+(K-D)}+x_{i+2\lambda+(K-D)}+x_{i+2\lambda+2(K-D)}+x_{i+3\lambda+2(K-D)}+\dots+x_{i+(p-1)\lambda+(p-1)(K-D)}+x_{i+p\lambda+(p-1)(K-D)}$. In the code symbol $C_{i}$, for $R_{k}$ all other symbols are the antidotes. Thus the $k^{th}$ receiver can decode its required message symbol $x_{k}$. \\ 

\noindent
\emph{Case (ii).} $k \in \mathcal{A}_{2}$:\\
  The $k^{th}$  receiver wants to decode $x_{k}$. Let \mbox{$i=k$ mod $(K-D+\lambda)$}. From the construction of the code for every $k$ in the given range the code book consists of the code symbol $C_{i}$ =  $x_{i}+x_{i+\lambda}+x_{i+\lambda+(K-D)}+x_{i+2\cdot\lambda+(K-D)}+x_{i+2\lambda+2(K-D)}+x_{i+3\lambda+2\cdot(K-D)}+\dots+x_{i+(p-1)\lambda+(p-1)(K-D)}+x_{i+p\lambda+(p-1)(K-D)}$. In the above code symbol for receiver $R_{k}$, the message symbol present just before $x_{k}$ (i.e. $x_{k-\lambda}$) is the interference to $R_{k}$ and all other message symbols are the antidotes. We will cancel the interference by using other code symbols. By adding the code symbols $$\sum_{s=0}^{m-1} C_{i-s\lambda}$$ we can replace the interference $x_{k-\lambda}$ with the antidote  $x_{k-K+D}$ and thus $x_{k}$ can be decoded.\\
The number of code symbols transmitted is equal to the number of values that $i$ takes and is equal to $K-D$. The capacity achieved by this code is $\frac{1}{K-D}$ the proof for which is along the same line as in the proof  Theorem \ref{thm1}.

\end{proof}

\small{By using the proposed code, \mbox{$(K-D)p$} receivers can decode their wanted messages by using just one transmitted symbol. Remaining receivers use \mbox{$m$} transmissions to decode their wanted messages.\\
\begin{example}
\label{ex8}
Let $K=24,\ D=19,\ \lambda=1$. Then $C=\frac{1}{K-D}=\frac{1}{5}$. Theorem \ref{thm8} gives the code\\
$\mathfrak{C}_{8}=\{x_{1}+x_{2}+ x_{7}+x_{8}+x_{13}+x_{14}+x_{19}+x_{20},\\ 
~~~~ x_{2}+x_{3}+ x_{8}+x_{9}+x_{14}+x_{15}+x_{20}+x_{21},\\ 
~~~~ x_{3}+x_{4}+ x_{9}+x_{10}+x_{15}+x_{16}+x_{21}+x_{22},\\ 
~~~~ x_{4}+x_{5}+ x_{10}+x_{11}+x_{16}+x_{17}+x_{22}+x_{23},\\ 
~~~~ x_{5}+x_{6}+ x_{11}+x_{12}+x_{17}+x_{18}+x_{23}+x_{24}\}$.\\
\noindent
The code is given by $\mathfrak{C}_{8}=\underline{x}L_{8}$ where $\underline{x}$=$[x_{1},x_{2}, \dots ,x_{24}]$ and $L_{8}$ is a $24 \times 5$ matrix which is given below.\\
\begin{center}
{\small
$L_{8} = \left[\begin{array}{*{20}c}
   1 & 0 & 0 & 0 & 0\\
   1 & 1 & 0 & 0 & 0\\
   0 & 1 & 1 & 0 & 0\\
   0 & 0 & 1 & 1 & 0\\
   0 & 0 & 0 & 1 & 1\\
   0 & 0 & 0 & 0 & 1\\
   1 & 0 & 0 & 0 & 0\\
   1 & 1 & 0 & 0 & 0\\
   0 & 1 & 1 & 0 & 0\\
   0 & 0 & 1 & 1 & 0\\
   0 & 0 & 0 & 1 & 1\\
   0 & 0 & 0 & 0 & 1\\
   1 & 0 & 0 & 0 & 0\\
   1 & 1 & 0 & 0 & 0\\
   0 & 1 & 1 & 0 & 0\\
   0 & 0 & 1 & 1 & 0\\
   0 & 0 & 0 & 1 & 1\\
   0 & 0 & 0 & 0 & 1\\
   1 & 0 & 0 & 0 & 0\\
   1 & 1 & 0 & 0 & 0\\
   0 & 1 & 1 & 0 & 0\\
   0 & 0 & 1 & 1 & 0\\
   0 & 0 & 0 & 1 & 1\\
   0 & 0 & 0 & 0 & 1\\
  \end{array}\right]$}
\end{center}
\end{example}

\noindent
{\bf Case IX:}
There is an integer $\lambda$ such that $D$ divides $K+\lambda$ and $\lambda$ divides $D$ and
\begin{equation}
\label{antidote9}
{\cal K}_k=\left\{
                \begin{array}{ll}
                  \{x_{k+D}\}$, if $\ k\leq K-2D+\lambda\\
                  \{x_{k+\lambda}, x_{k+2\lambda},...,x_{k+D}\}$, if $K-2D+\lambda<k\leq K
                  \end{array}
              \right.
\end{equation}
\begin{theorem}
\label{thm9}
If  $D$ divides $K+\lambda$ and $\lambda$ divides $D$ and the antidote pattern is as in \eqref{antidote9}, then the optimal scalar linear code proposed is given by\\
 $\mathfrak{C}=\{{x_{i+(j-1)D}+x_{i+jD}}|i = 1,2,\dots,D, ~~j = 1,2,\dots,n-2\}\\ 
~~~~~~~~~ \cup \{x_{K-2D+1+\lambda+i'}+x_{K-D+1+i'}+x_{K-\lambda+1+i' mod \lambda}| \\
~~~~~~~~~~~~~~~~~~~~~~~~~~~~~~~~~~ i' = \{0,1,2,\dots,p-1\}\}$\\
where $\frac{K+\lambda}{D}=n(>2)$, $p=K$ mod $D$ = $D-\lambda$.
\end{theorem}

\begin{proof}
We will establish that by using the proposed code every receiver can decode its wanted message.\\
 
\noindent
\emph{Case (i).} $k\leq {K-2D+\lambda}$:\\
The $k^{th}$  receiver wants to decode $x_{k}$. The code book consists of the code symbol of the form $x_{i+(j-1)D}+x_{i+jD}$. As $i$ and $j$ runs, $x_{i+(j-1)D}+x_{i+jD}$ runs from $x_{1}+x_{1+D}$ to $x_{K-D-\lambda}+x_{K-\lambda}$. Therefore for every {$k\leq {K-2D+\lambda}$}, the code book consists of the symbol $x_{k} + x_{k+D}$. Since the symbol $x_{k+D}$ is the antidote of $k^{th}$ receiver, $k^{th}$ receiver can decode its required symbol $x_{k}$. \\

\noindent
\emph{Case (ii).} ${K-2D+\lambda}<k \leq{K-D-\lambda}$:\\
The receiver $R_{k}$ wants to decode $x_{k}=x_{K-2D+1+\lambda+i'}$. From the construction of the code, the code book consists of the code symbol $\mathcal{C}_{i'}=x_{K-2D+1+\lambda+i'}+x_{K-D+1+i'}+ x_{K-\lambda+1+i' mod \lambda}|\ i' = \{0,1,2,\dots,p-\lambda-1\}$. By adding $\mathcal{C}_{i'}$ and $\mathcal{C}_{i'+\lambda}$, we get  $S=x_{K-2D+1+\lambda+i'}+x_{K-D+1+i'}+x_{K-2D+1+2\lambda+i'}+x_{K-D+1+\lambda+i'}$. For the receiver $R_{k}$ every other message symbol in $S$ is antidote. Thus receiver $R_{k}$ can decode $x_{k}$.\\

\noindent
\emph{Case (iii).} ${K-D-\lambda}<k \leq{K-D}$:\\
The receiver $R_{k}$ wants to decode $x_{k}=x_{K-2D+1+\lambda+i'}$. From the construction of the code, the code book consists of the code symbol $\mathcal{C}_{i'}=x_{K-2D+1+\lambda+i'}+x_{K-D+1+i'}+x_{K-\lambda+1+i' mod \lambda}|\ i'=\{p-\lambda,\dots,p-1\}$. For the receiver $R_{k}$ for ${K-D+\lambda}<k \leq{K-D}$,  every other message symbol is the code symbol $x_{K-2D+1+\lambda+i'}+x_{K-D+1+i'}+x_{K-\lambda+1+i' mod \lambda}$ is antidote. Thus receiver $R_{k}$ can decode $x_{k}$.\\

\noindent
\emph{Case (iv).} ${K-D}<k \leq{K-\lambda}$:\\
The receiver $R_{k}$ wants to decode $x_{k}=x_{K-D+1+i'}$. From the construction of the code, the code book consists of the code symbol $x_{K-2D+1+\lambda+i'}+x_{K-D+1+i'}+x_{K-\lambda+1+i' mod \lambda}|\ i' = \{0,1,2,\dots,p-1\}$. 
For the required message symbol $x_{K-D+1+i'}$, message symbol $x_{K-\lambda+1+i' mod \lambda}$ is antidote and  $x_{K-2D+1+\lambda+i'}$ is interference. We will cancel the interference by using other code symbols. The code book also consists of code symbols of the form $\mathcal{C}_{i,j}={x_{i+(j-1)D}+x_{i+jD}}$. By adding the code symbols $$\sum_{j=1}^{n-2} C_{i'+1,\ j}$$ we get $x_{i'+1}+x_{K-2D+1+\lambda+i'}$. The message symbol $x_{i'+1}$ is antidote for the receiver $R_{k}$. Thus by using $x_{i'+1}+x_{K-2D+1+\lambda+i'}$, the interference symbol can be replaced with the antidote symbol and thus $x_{k}$ can be decoded.\\

\noindent
\emph{Case (v).} $k>K-\lambda$:\\
The receiver $R_{k}$ wants to decode $x_{k}=x_{K-\lambda+1+i' mod \lambda}$. From construction of the code the code book consists of the code symbol $x_{K-2D+1+\lambda+i'}+x_{K-D+1+i'}+x_{K-\lambda+1+i' mod \lambda}$.  For the message symbol $x_{K-\lambda+1+i' mod \lambda}$, the message symbols $x_{K-2D+1+\lambda+i'}$ and $x_{K-D+1+i'}$ are interference. We will cancel the interference by using other code symbols. The code book also consists of code symbols of the form $\mathcal{C}_{i,j}={x_{i+(j-1)D}+x_{i+jD}}$. By adding the code symbols $$\sum_{j=1}^{n-2} C_{(i' mod \lambda)+1,\ j}$$ we get $x_{(i' mod \lambda)+1}+x_{K-2D+1+\lambda+i' mod \lambda}$. The message symbol $x_{(i' mod \lambda)+1}$ is antidote for the receiver $R_{k}$. Similarly by adding the code symbols $$\sum_{j=1}^{n-2} C_{(i' mod \lambda)+1+D-\lambda,\ j}$$ we get $x_{(i' mod \lambda)+1+D-\lambda}+x_{K-D+1+i' mod \lambda}$. The message symbol $x_{(i' mod \lambda)+1+D-\lambda}$ is antidote for the receiver $R_{k}$. Thus by using $x_{(i' mod \lambda)+1}+x_{K-2D+1+\lambda+i' mod \lambda}$ and  $x_{(i' mod \lambda)+1+D-\lambda}+x_{K-D+1+i' mod \lambda}$, the interference symbols can be replaced with the antidote symbols and thus $x_{k}$ can be decoded.\\

The number of code symbols transmitted is equal to $p$ more than the product of the number of values that $i$ takes and the number of values that $j$ takes. Number of code symbols=$D(n-2)+p=D(\frac{K+\lambda}{D}$-$2)$+$D-\lambda$=$K$-$D$. The rest of the proof of optimality of the length is same as that of Theorem \ref{thm1}.

\end{proof}

\small{By using the proposed code, \mbox{$K-2D+2\lambda$} receivers can decode their wanted messages by using just one transmitted symbol. \mbox{$D-2\lambda$} receivers decode their wanted messages by using 2 transmissions each. Remaining receivers use atmost $2n-3$ transmissions each.\\ 
\begin{example}
\label{ex9}
Let $K=19,\ D=5,\ \lambda=1$. Then $C=\frac{1}{K-D}=\frac{1}{14}$. Theorem \ref{thm9} gives the code\\
$\mathfrak{C}_{9}=\{x_{1}+x_{6},\ x_{6}+x_{11},\ x_{11}+x_{15}+x_{19},\\ x_{2}+x_{7},\ x_{7}+x_{12},\ x_{12}+x_{16}+x_{19},\\ x_{3}+x_{8},\ x_{8}+x_{13},\ x_{13}+x_{17}+x_{19},\\ x_{4}+x_{9},\ x_{9}+x_{14},\ x_{14}+x_{18}+x_{19}\}\\ x_{5}+x_{10},\ x_{10}+x_{15},$\\
\noindent
The code is given by $\mathfrak{C}_{9}=\underline{x}L_{9}$ where $\underline{x}$=$[x_{1},x_{2}, \dots ,x_{19}]$ and $L_{9}$ is the $19 \times 14$ matrix which is given below.
\begin{center}
{\scriptsize
$$ L_{9}= \left[\begin{array}{*{20}c}
   1 & 0 & 0 & 0 & 0 & 0 & 0 & 0 & 0 & 0 & 0 & 0 & 0 & 0\\
   0 & 1 & 0 & 0 & 0 & 0 & 0 & 0 & 0 & 0 & 0 & 0 & 0 & 0\\
   0 & 0 & 1 & 0 & 0 & 0 & 0 & 0 & 0 & 0 & 0 & 0 & 0 & 0\\
   0 & 0 & 0 & 1 & 0 & 0 & 0 & 0 & 0 & 0 & 0 & 0 & 0 & 0\\
   0 & 0 & 0 & 0 & 1 & 0 & 0 & 0 & 0 & 0 & 0 & 0 & 0 & 0\\
   1 & 0 & 0 & 0 & 0 & 1 & 0 & 0 & 0 & 0 & 0 & 0 & 0 & 0\\
   0 & 1 & 0 & 0 & 0 & 0 & 1 & 0 & 0 & 0 & 0 & 0 & 0 & 0\\
   0 & 0 & 1 & 0 & 0 & 0 & 0 & 1 & 0 & 0 & 0 & 0 & 0 & 0\\
   0 & 0 & 0 & 1 & 0 & 0 & 0 & 0 & 1 & 0 & 0 & 0 & 0 & 0\\
   0 & 0 & 0 & 0 & 1 & 0 & 0 & 0 & 0 & 1 & 0 & 0 & 0 & 0\\
   0 & 0 & 0 & 0 & 0 & 1 & 0 & 0 & 0 & 0 & 1 & 0 & 0 & 0\\
   0 & 0 & 0 & 0 & 0 & 0 & 1 & 0 & 0 & 0 & 0 & 1 & 0 & 0\\
   0 & 0 & 0 & 0 & 0 & 0 & 0 & 1 & 0 & 0 & 0 & 0 & 1 & 0\\
   0 & 0 & 0 & 0 & 0 & 0 & 0 & 0 & 1 & 0 & 0 & 0 & 0 & 1\\
   0 & 0 & 0 & 0 & 0 & 0 & 0 & 0 & 0 & 1 & 1 & 0 & 0 & 0\\
   0 & 0 & 0 & 0 & 0 & 0 & 0 & 0 & 0 & 0 & 0 & 1 & 0 & 0\\
   0 & 0 & 0 & 0 & 0 & 0 & 0 & 0 & 0 & 0 & 0 & 0 & 1 & 0\\
   0 & 0 & 0 & 0 & 0 & 0 & 0 & 0 & 0 & 0 & 0 & 0 & 0 & 1\\
   0 & 0 & 0 & 0 & 0 & 0 & 0 & 0 & 0 & 0 & 1 & 1 & 1 & 1\\
   \end{array}\right]$$
}
\end{center}
\end{example}
~ \\

\noindent
{\bf Case X:} 
There is an integer $\lambda$ such that $K-D$ divides $K+\lambda$ and $\lambda$ divides $K-D$ and the antidote pattern is given by \eqref{antidote1}.

\begin{theorem}
\label{thm10}
If $K-D$ divides $K+\lambda$ and $\lambda$ divides $K-D$ and the antidote pattern is given by \eqref{antidote1}, then the proposed optimal scalar linear code is\\ 
$\mathfrak{C}=\{x_{k}+x_{k+m}+x_{k+2m}+ \dots+x_{k+(q-1)m}+x_{k+(q-1)m+\lambda}\\
~~~~~~+x_{k+(q-1)m+2\lambda}+ \dots + x_{k+(q-1)m+(s-2)\lambda}| k = 1,2,\dots,\lambda\} \\ 
\cup \{x_{k}+x_{k+m}+x_{k+2m}\dots+x_{k+(q-2)m}+x_{k+(q-1)m-\lambda}|\\
~~~~~~~~~~~~~~~~~~~~~~~~~~~~~~~~~~~~~~~~~ k = \{\lambda+1,\lambda+2,\dots,p\}\} \\ 
\cup \{x_{k}+x_{k+m}+x_{k+2m}+ \dots+x_{k+(q-2)m}+x_{k+(q-2)m+\lambda}+x_{k+(q-2)m+2\lambda}+\dots + x_{k+(q-2)m+(s-1)\lambda}|\ k = p+1,p+2,\dots,m\}$\\ 
where $K-D=m$, $K-D-\lambda=p$, $\frac{K+\lambda}{K-D}=q$  and \mbox{$\frac{K-D}{\lambda}=s$}.
\end{theorem}

\begin{proof}
We will prove that by using this code every receiver can decode its wanted message.\\

\noindent
\emph{Case (i).} $k\in \{1,2,\dots,\lambda\}$:\\
 The $k^{th}$  receiver wants to decode $x_{k}$.  From the construction of the code for every \mbox{$k\in \{1,2,\dots,\lambda\}$} the $k^{th}$ code symbol $\mathcal{C}_{k}$ and $(k+p)^{th}$ code symbol $\mathcal{C}_{k+p}$ have same message symbols in last $s-1$ positions. By adding $\mathcal{C}_{k}$ and $\mathcal{C}_{k+p}$, we get  $\mathcal{C}_{k}=x_{k}+x_{k+m}+x_{k+2m}+ \dots+x_{k+(q-2)m}+x_{k+p}+x_{k+m+p}+x_{k+2m+p}+\dots+x_{k+(q-2)m+p}$. In this sum every other message symbol is antidote of $R_{k}$. Thus $R_{k}$ can decode its wanted message $x_{k}$, $\forall$ $k\in \{1,2,\dots,\lambda\}$.\\
\noindent
\emph{Case (ii).} $k\in \{\lambda+1,\lambda+2,\dots,K-D\}$:\\
From the construction of the code for every $k\in \{\lambda+1,\lambda+2,\dots,K-D\}$, every other message symbol in the $k^{th}$ code symbol is antidote of $R_{k}$, thus $R_{k}$ can decode its wanted message $x_{k}$.\\
\noindent
\emph{Case (iii).} $k\in \{K-D+1,K-D+2,\dots,D+2\lambda\}$:\\
 Let $i$=1+$(k-1)$ mod $m$. From the construction of the code for every $k\in \{K-D+1,K-D+2,\dots,D+\lambda\}$, the $k^{th}$ message symbol is present in the $i^{th}$ code symbol. In the $i^{th}$ code symbol, every other message symbol present is antidote of  $R_{k}$, thus $R_{k}$ can decode its wanted message $x_{k}$.\\
\noindent
\emph{Case (iv).} $k\in \{D+2\lambda+1,D+2\lambda+2,\dots,K\}$:\\
Let $i$=1+$(k-1)$ mod $\lambda$. From the construction of the code, the code book consists of the code symbol $\mathcal{C}_{i}=x_{i}+x_{i+m}+x_{i+2m}+ \dots+x_{i+(q-1)m}+x_{i+(q-1)m+\lambda}+x_{i+(q-1)m+2\lambda}+ \dots + x_{i+(q-1)m+(s-2)\lambda}$ where $i = \{1,2,\dots,\lambda\}$.\\

Let $t=\lceil\frac{k-D-2\lambda}{\lambda}\rceil$. The $t$ message symbols present before $x_{k}=x_{i+(q-1)m+t\lambda}$ in the $i^{th}$ code symbol contribute to the interference of $R_{k}$.  By adding the code symbols $$\sum_{j=0}^{t} C_{i+j\lambda}$$the interference symbols for the receiver $R_{k}$ can be canceled and thus receiver $R_{k}$ can decode $x_{k}$. \\

The proof for optimality follows from the fact that the number of code symbols transmitted is equal to $\lambda+(p-\lambda)+(m-p)=m= K-D$ and the reasoning along the same line as the optimality for the code in Theorem \ref{thm1}. 
\end{proof}
By using the proposed code, \mbox{$D+\lambda$} receivers can decode their wanted messages by using just one transmitted code symbol. The remaining receivers decode their wanted messages by using atmost \mbox{$\frac{K-D-\lambda}{\lambda}$} transmissions each.
\begin{example}
\label{ex10}
Let $K=28,\ D=18,\ \lambda=2$.Then $C=\frac{1}{K-D}=\frac{1}{10}$. Theorem \ref{thm10} gives the code\\
$\mathfrak{C}_{10}= \{x_{1}+x_{11}+x_{21}+x_{23}+x_{25}+x_{27},\\ 
~~~~~~~~ x_{2}+x_{12}+x_{22}+x_{24}+x_{26}+x_{28},\\ 
~~~~~~~~ x_{3}+x_{13}+x_{21},~~~ x_{4}+x_{14}+x_{22}, ~~~ x_{5}+x_{15}+x_{23},\\
~~~~~~~~ x_{6}+x_{16}+x_{24}, ~~~ x_{7}+x_{17}+x_{25}, ~~~  x_{8}+x_{18}+x_{26},\\ 
~~~~~~~~ X_{9}+x_{19}+x_{21}+x_{23}+x_{25}+x_{27},\\ 
~~~~~~~~ x_{10}+x_{20}+x_{22}+x_{24}+x_{26}+x_{28}\}.$

The code is given by $\mathfrak{C}_{10}=\underline{x}L_{10}$ where $\underline{x}$=$[x_{1},x_{2}, \dots ,x_{28}]$ and $L_{10}$ is the  $28 \times 10$ matrix which is given below.\\
\begin{center}
{\scriptsize
$L_{10} = \left[\begin{array}{*{20}c}
   1 & 0 & 0 & 0 & 0 & 0 & 0 & 0 & 0 & 0\\
   0 & 1 & 0 & 0 & 0 & 0 & 0 & 0 & 0 & 0\\
   0 & 0 & 1 & 0 & 0 & 0 & 0 & 0 & 0 & 0\\
   0 & 0 & 0 & 1 & 0 & 0 & 0 & 0 & 0 & 0\\
   0 & 0 & 0 & 0 & 1 & 0 & 0 & 0 & 0 & 0\\
   0 & 0 & 0 & 0 & 0 & 1 & 0 & 0 & 0 & 0\\
   0 & 0 & 0 & 0 & 0 & 0 & 1 & 0 & 0 & 0\\
   0 & 0 & 0 & 0 & 0 & 0 & 0 & 1 & 0 & 0\\
   0 & 0 & 0 & 0 & 0 & 0 & 0 & 0 & 1 & 0\\
   0 & 0 & 0 & 0 & 0 & 0 & 0 & 0 & 0 & 1\\
   1 & 0 & 0 & 0 & 0 & 0 & 0 & 0 & 0 & 0\\
   0 & 1 & 0 & 0 & 0 & 0 & 0 & 0 & 0 & 0\\
   0 & 0 & 1 & 0 & 0 & 0 & 0 & 0 & 0 & 0\\
   0 & 0 & 0 & 1 & 0 & 0 & 0 & 0 & 0 & 0\\
   0 & 0 & 0 & 0 & 1 & 0 & 0 & 0 & 0 & 0\\
   0 & 0 & 0 & 0 & 0 & 1 & 0 & 0 & 0 & 0\\
   0 & 0 & 0 & 0 & 0 & 0 & 1 & 0 & 0 & 0\\
   0 & 0 & 0 & 0 & 0 & 0 & 0 & 1 & 0 & 0\\
   0 & 0 & 0 & 0 & 0 & 0 & 0 & 0 & 1 & 0\\
   0 & 0 & 0 & 0 & 0 & 0 & 0 & 0 & 0 & 1\\
   1 & 0 & 1 & 0 & 0 & 0 & 0 & 0 & 1 & 0\\
   0 & 1 & 0 & 1 & 0 & 0 & 0 & 0 & 0 & 1\\
   1 & 0 & 0 & 0 & 1 & 0 & 0 & 0 & 1 & 0\\
   0 & 1 & 0 & 0 & 0 & 1 & 0 & 0 & 0 & 1\\
   1 & 0 & 0 & 0 & 0 & 0 & 1 & 0 & 1 & 0\\
   0 & 1 & 0 & 0 & 0 & 0 & 0 & 1 & 0 & 1\\
   1 & 0 & 0 & 0 & 0 & 0 & 0 & 0 & 1 & 0\\
   0 & 1 & 0 & 0 & 0 & 0 & 0 & 0 & 0 & 1\\
  \end{array}\right]$}
\end{center}
\end{example}
~ \\

\section{conclusion}
In this paper optimal length scalar linear index codes have been proposed and studied for ten cases of symmetric multiple unicast index coding problems with antidotes smaller than those that are reported in \cite{MCJ}. 
In fact, by using proposed   codes we have proved the capacity for some specific antidote patterns considered in this manuscript. However, for arbitrary $K$ and $D$, the design of optimal linear codes to achieve capacity is open.

\end{document}